%% file: main.tex
\documentclass[10pt,conference]{IEEEtran}
\IEEEoverridecommandlockouts
\usepackage{graphicx,cite,calc,xcolor,subfigmat,amssymb,amsmath,mathrsfs,dsfont,epstopdf}
\usepackage[normalem]{ulem}
\usepackage{ulem}
\usepackage{graphicx} 
\usepackage[letterpaper, left=0.68in, right=0.68in, top=0.75in, bottom=1in]{geometry}

\usepackage[switch]{lineno}

\usepackage{cite}
\usepackage[linesnumbered,ruled,vlined]{algorithm2e}
\usepackage{graphicx, amsmath, amssymb} 

\usepackage{algpseudocode}
\usepackage{mathtools}
\usepackage{algorithmicx}

\usepackage{dsfont}
\usepackage{bbm}
\usepackage{soul}
\usepackage{xcolor}
\usepackage{amsmath,amsthm}
\usepackage{amsmath}

\usepackage{float}
\usepackage[T1]{fontenc}
\usepackage{soul} 
\usepackage{xcolor}
\usepackage{optidef}
\usepackage{listings}
\usepackage{enumitem}
\newtheorem{theorem}{Theorem}

\newtheorem{prop}{Proposition}

\newtheorem{lemma}{Lemma}
\newtheorem{corollary}{Corollary}[theorem]

\makeatletter
\newcommand*{\rom}[1]{\expandafter\@slowromancap\romannumeral #1@}
\makeatother
\usepackage {amsmath}
\definecolor{highlightcolor}{RGB}{255,255,0}

\usepackage{newtxtext,newtxmath}
\title {Enhancing Sustainability in HAPS-Assisted 6G Networks: Load Estimation Aware Cell Switching}

\begin{document}

 \author{\rm{Maryam Salamatmoghadasi}, \rm{Metin Ozturk}, \rm{Halim Yanikomeroglu}
 \thanks{This research has been sponsored in part by the NSERC Create program entitled TrustCAV and in part by The Scientific and Technological Research Council of Türkiye (TUBITAK).
 
Maryam Salamatmoghadasi and Halim Yanikomeroglu are with Non-Terrestrial Networks Lab, Department of Systems and Computer Engineering, Carleton University, Ottawa, ON K1S5B6, Canada. Metin Ozturk is with Electrical and Electronics Engineering, Ankara Yildirim Beyazit University, Ankara, 06010, Turkiye.}}

\maketitle
\begin{abstract}
This study introduces and addresses the critical challenge of traffic load estimation in cell switching within vertical heterogeneous networks (vHetNets).
The effectiveness of cell switching is significantly limited by the lack of accurate traffic load data for small base stations (SBSs) in sleep mode, making many load-dependent energy-saving approaches impractical, as they assume perfect knowledge of traffic loads—an assumption that is unrealistic when SBSs are inactive.
In other words, when SBSs are in sleep mode, their traffic loads cannot be directly known and can only be estimated, inevitably with corresponding errors.
Rather than proposing a new switching algorithm, we focus on eliminating this foundational barrier by exploring effective prediction techniques. A novel vHetNet model is considered, integrating a high-altitude platform station (HAPS) as a super macro base station (SMBS).
We investigate both spatial and temporal load estimation approaches, including three spatial interpolation schemes—random neighboring selection, distance-based selection, and multi-level clustering (MLC)—alongside a temporal deep learning method based on long short-term memory (LSTM) networks. 
Using a real-world dataset for empirical validation, our results show that both spatial and temporal methods significantly improve estimation accuracy, with the MLC and LSTM approaches demonstrating particularly strong performance.

\end{abstract}
\section{Introduction}

With the advent of sixth-generation (6G) cellular networks, expected to support an unprecedented number of devices per square kilometer, the pursuit of enhanced connectivity and efficiency becomes imperative. 
However, this expansion faces significant challenges, including a notable increase in energy consumption within radio access networks~(RANs), with base stations~(BSs), particularly small BSs~(SBSs), being major contributors~\cite{3333333}.
In this regard, strategically deactivating BSs or putting them into sleep mode during low activity periods emerges as a feasible solution for enhancing energy efficiency and network sustainability. 
Nonetheless, implementing effective cell switching strategies is severely constrained by the absence of precise traffic load data of sleeping SBSs at the next time slot, which is critical for informed offloading and energy optimization decisions. Nonetheless, most existing studies rely on an unrealistic assumption of perfect traffic load knowledge for sleeping SBSs~\cite{10304250,HETS2023P,ELAA2022JR,EOMK2017JR,Metin_VFA_CellSwitch}.

Overcoming this hurdle by accurately estimating the traffic loads of sleeping SBSs is crucial, bridging the gap between theory and practice and unlocking the full potential of cell switching for significantly improving sustainability in vertical heterogeneous networks (vHetNets), which integrates satellite, high altitude platform station (HAPS), uncrewed aerial vehicle (UAV), and terrestrial BSs to ensure comprehensive global coverage~\cite{10938203}. HAPS, in particular, provides wide-area, line-of-sight (LoS) coverage and large backhaul capacity, making them ideal super macro base station (SMBS) candidates in support of terrestrial networks. This motivates our use of HAPS in the vHetNet model.
\textit{Therefore, in this study, instead of taking a small step forward and developing another advanced cell switching algorithm, we concentrate on removing the barrier between the state-of-the-art advanced techniques and their real-life implementations, because the majority of the existing algorithms in the literature are impractical due to the traffic load estimation problem.}

Research on cell switching has extensively focused on more efficient BS deactivation strategies to reduce network energy consumption. For instance, the study in~\cite{10304250} examined a vHetNet model with a HAPS functioning as an SMBS, alongside a macro BS (MBS) and several SBSs, aiming to optimize sleep mode management of SBSs and utilize HAPS-SMBS capabilities to lower energy use while maintaining user quality-of-service (QoS).
Similarly, research in~\cite{HETS2023P} addressed cell switching challenges with HAPS-SMBS, specifically targeting traffic offloading from deactivated BSs using a sorting algorithm that prioritizes BSs with lower traffic loads for deactivation.
A tiered sleep mode system that adjusts sleep depth according to device activity was proposed in~\cite{ELAA2022JR}, featuring decentralized control for scalability and efficiency.
Additionally, the study in~\cite{EOMK2017JR} considered the control data separated architecture (CDSA) and implemented a genetic algorithm to optimize energy savings in HetNets by managing user associations and BS deactivation through deterministic algorithms.
Furthermore, the work in~\cite{Metin_VFA_CellSwitch} introduced a value function approximation (VFA)-based reinforcement learning (RL) algorithm for cell switching in ultra-dense networks, demonstrating scalable energy savings while maintaining QoS.

However, the challenge of estimating the traffic load for sleeping BS remains unaddressed, limiting the practical application of these studies and their contribution to the achievement of the sustainability goals of 6G networks as outlined in the 6G framework of the International Telecommunication Union (ITU)~\cite{itu_vision_june_23}.
Addressing this critical gap, our work presents a pioneering effort to explore both the spatial and temporal dimensions of traffic load estimation for sleeping SBSs—an aspect largely overlooked in the current literature. We go beyond conventional assumptions of perfect traffic knowledge by systematically evaluating the effects of estimation errors on cell switching performance in a realistic vHetNet environment~\cite{10304250}. On the spatial side, we investigate three interpolation techniques: random neighboring selection, distance-based selection, and clustering-based selection. 
On the temporal side, we introduce a deep learning-based solution using long short-term memory (LSTM) networks—marking a novel application of LSTM to estimate SBS traffic loads for energy-aware cell switching in vHetNets.
All proposed methods are validated using a real-world call detail record (CDR) dataset from Milan~\cite{DVN/EGZHFV_2015}, ensuring practical relevance and generalizability. The key contributions of this work are summarized as follows:

\begin{itemize} 
\item Addressing the challenge of traffic load estimation for sleeping SBSs in cell switching within vHetNets. \item Implementing and evaluating three spatial estimation methods—random neighboring selection, distance-based selection, and $k$-means clustering—along with a novel temporal prediction approach using LSTM networks. \item Developing a mathematical framework to analyze the behavior and accuracy of spatial interpolation techniques. \item Utilizing a real CDR dataset to empirically validate the spatial and temporal estimation schemes under realistic network traffic patterns. 
\end{itemize}

\section{System Model}\label{sec:model}
\subsection{Network Model}
Our study explores a vHetNet depicted in Fig. \ref{fig-1}, comprising a macro cell (MC) with a single MBS and $n \in \mathbb{N}$ SBSs.
Additionally, a HAPS-SMBS is integrated, potentially serving multiple MCs, into the network.
The primary function of SBSs is to deliver data services and address user-specific requirements, while MBS and HAPS-SMBS ensure consistent network coverage and manage control signals.
A key role of the HAPS-SMBS is to efficiently manage traffic offloading from SBSs during low-traffic periods, utilizing its extensive LoS and large capacity~\cite{9380673}. 

\begin{figure}[t]
\centering
\includegraphics[width = 7.cm]{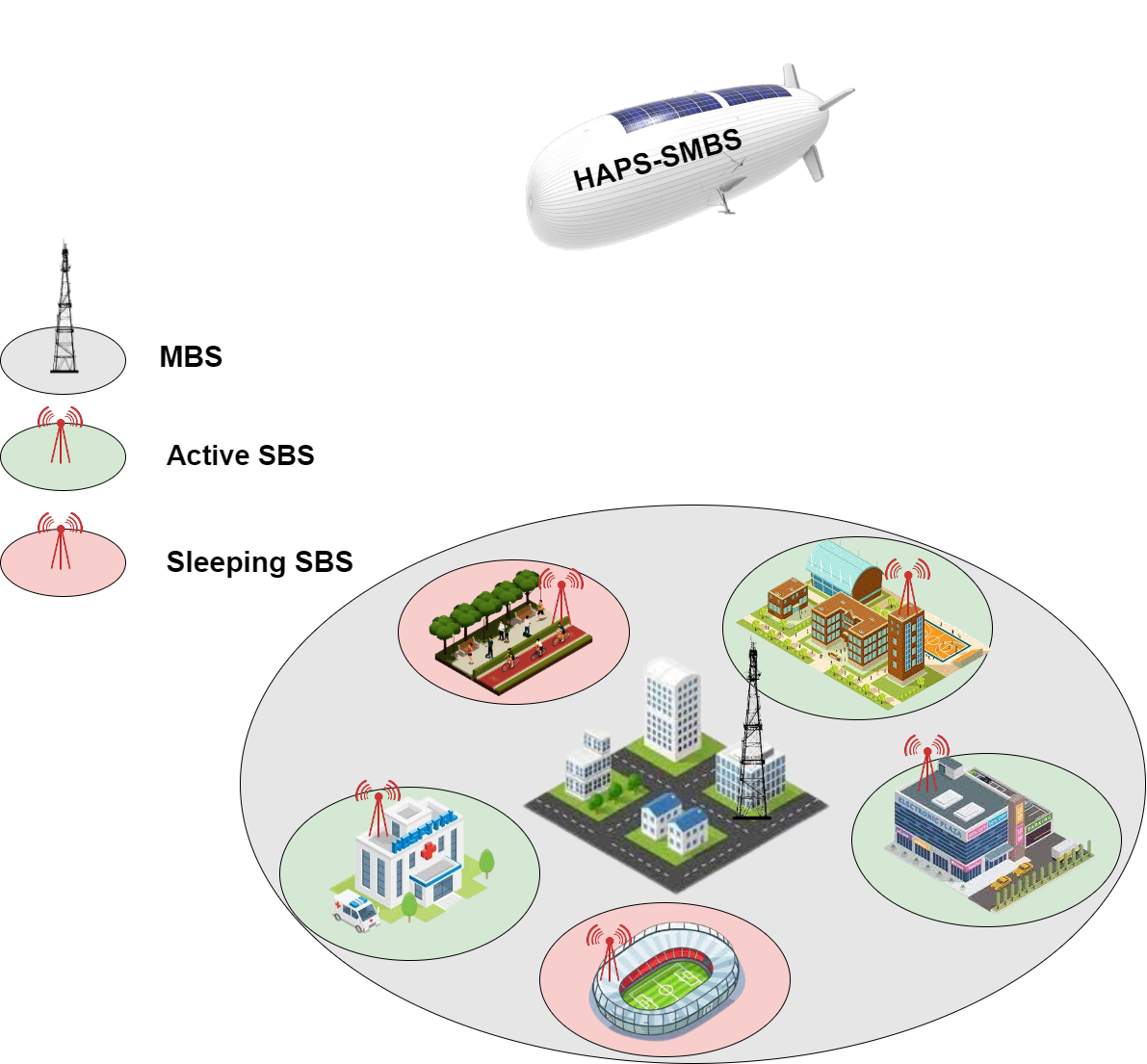}
\caption{A vHetNet model with an MBS, multiple SBSs, and a HAPS-SMBS.}
\label{fig-1}
\vspace{-1em}
\end{figure}

\subsection{Traffic Load Dependent Network Power Consumption}
The power consumption of each BS in the network is calculated based on the energy-aware radio and network technologies (EARTH) power consumption model. 
For the $i$-th BS, the power consumption at any given time, denoted as $P_i$, is expressed as \cite{7925662}
\begin{equation}
P_i  = \left\{ {\begin{array}{*{20}c}
   {\begin{array}{*{20}c}
   {P_{\text{o},i}  + \eta _i \lambda _i P_{\text{t},i} }, & {0 < \lambda _i  < 1,}  
\end{array}}  \\
   {\begin{array}{*{20}c}
   {P_{\text{s},i} },  \; \;   \; \; \; \; \; \; \;\; \; \; \; \; \; \; \; \; \; \; \; \; \; & {\lambda _i  = 0},  \\
\end{array}}  \\
\end{array}} \right.
\label{eq1}
\end{equation}
where $P_{\text{o},i}$ represents the operational circuit power consumption, $\eta_i$ is the power amplifier efficiency, $\lambda_i$ is the load factor (i.e., the normalized traffic load), $P_{\text{t},i}$ is the transmit power, and $P_{\text{s},i}$ is the power consumption in sleep mode. 
The total instantaneous power consumption of vHetNet, denoted as $P_\text{T}$, is given by
\begin{equation}
P_\text{T} = P_\text{H}  + P_\text{M}  + \sum\limits_{j = 1}^s {P_j }, 
\label{eq2}
\end{equation}
where $P_\text{H}$ and $P_\text{M}$ denote the power consumption of HAPS-SMBS and MBS at any given moment, respectively, which are calculated based on the $(0 < \lambda _i  < 1)$ case in \eqref{eq1} as HAPS-SMBS and MBS are always active in our modeling. 
Meanwhile, $P_j$ represents the power consumption of SBS $j$ and $s$ signifies the total number of SBSs within the network.

\subsection{Data Set and Data Processing}
To assess power consumption as defined in~\eqref{eq1}, we require the load factor $\lambda_i$ for each BS. 
Thus, we employ a real CDR data set from Telecom Italia~\cite{DVN/EGZHFV_2015} that captures user activity in Milan, partitioned into 10,000 grids of $235\times 235$ meters. This activity includes calls, texts, and internet usage recorded every 10 minutes over December 2013. We consolidate these activities into a single measure of traffic load per grid. Each SBS is then assigned a normalized traffic load from a randomly selected grid to simulate its corresponding cell activity.

\section{Multi-Dimensional Traffic Load Estimation: Spatial and Temporal Perspectives}\label{sec:method}
The model in~\eqref{eq1} reveals that the network power consumption and the accuracy of cell switching schemes are highly dependent on the traffic loads of BSs~\cite{10304250,Metin_VFA_CellSwitch,earth}, therefore any error in the traffic load estimation can significantly affect the cell switching performance.
We employ three spatial interpolation methodologies and a temporal modeling approach using LSTM to estimate the traffic load of sleeping SBSs. 
\subsection{Geographical Distance-Based Traffic Load Estimation}
This method considers the proximity of neighboring cells to estimate the traffic load of a sleeping SBS. It includes two sub-methods based on the presence of a weighting mechanism, which is developed to prioritize the effect of closer cells.
 
\subsubsection{Distance-Based without Weighting} In this approach, the traffic load of a sleeping SBS is estimated by averaging the traffic loads of its neighboring cells, arranged incrementally based on proximity. All neighboring cells contribute equally to the estimation, regardless of their distance. The estimated traffic load of SBS $j$, $\hat\lambda _j$, is calculated as~\cite{maryamarxive}
\begin{equation}
    \hat\lambda _j  = \frac{1}{N}\sum\limits_{a = 1}^N {\lambda _a,} 
\label{eq14}
\end{equation}
where ${{\lambda_a}}$ represents the traffic load of cell $a$, and $N$ is defined as the number of cells included in the estimation process.
Note that usually $N\gg s$, as $N$ encompasses all the cells available for the estimation process, while $s$ is the number of SBSs within a single vHetNet with a single MC.
  
\subsubsection{Distance-Based with Weighting} This method refines the previous approach by assigning different weights to neighboring cells based on their distance from the sleeping SBS. The closer a cell is, the more it influences the estimated traffic load. The weighted traffic load, $\hat\lambda _j$, is calculated as~\cite{maryamarxive}
\begin{equation}
\hat \lambda _j  = {\sum\limits_{a = 1}^N {\lambda _a  \times w_{j,a} } }\Big/{{\sum\limits_{a = 1}^N {w_{j,a} } }},
\label{eq15}
\end{equation}
where the weighting factor, ${w_{j,a}}$, is defined as
\begin{equation} 
w_{j,a}  = \frac{{d_{\max } }}{{d_{j,a}^n }},\quad
n \in \mathbb{R}^+, 
\label{eq16}
\end{equation}
where ${d_{\max }}$ is the maximum distance between the sleeping SBS and its neighboring cells included in the estimation, and ${d_{j,a}}$ is the distance between the sleeping SBS ${j}$ and the neighboring SBS ${a}$.
\begin{theorem}
The error in estimating the traffic load of a sleeping SBS decreases as the exponent $n$, representing the power of the distance between the sleeping SBS and its neighboring cells, increases.
\label{theorem2}
\end{theorem}
\begin{proof}
Substituting~\eqref{eq16} into \eqref{eq15}, we get 
\begin{equation}\label{eq:simplified}
    \hat \lambda _j  = \frac{{\sum\limits_{a = 1}^N {\lambda _a  \times\frac{{d_{\max } }}{{d_{j,a}^n }} } }}{{\sum\limits_{a = 1}^N {\frac{{d_{\max } }}{{d_{j,a}^n}}}}}= \frac{{\sum\limits_{a=1}^N \lambda_a \frac{{1}}{{d_{j,a}^n}}}}{{\sum\limits_{a=1}^N \frac{{1}}{{d_{j,a}^n}}}}.
\end{equation}
The $n$-dependent term of~\eqref{eq:simplified}, $\sum_{a=1}^N \dfrac{{1}}{{d_{j,a}^n}}$, is increasingly dominated by the terms with the smallest values of $d^n_{j,a}$ as $n$ increases, that is, the weighting factor, $w_{i,j}=\dfrac{{d_{\max}}}{{d_{j,a}^n}}$, heavily favors smaller distances and diminishes the influence of larger distances.
As $n$ goes to infinity, i.e., $\lim_{n \to \infty}\sum_{a=1}^N{\frac{{d_{\max}}}{{d_{j,a}^n}}}$ , the effect of larger distances (i.e., $d_{j,a}>1$) on the estimation approaches to zero, while smaller distances (i.e., $d_{j,a}<1$) have significant effects, such that 
\begin{equation}\label{eq:limit_to_infinity}
    \lim_{n \to \infty}\frac{{1}}{{d_{j,a}^n}} =\begin{cases}
        0, & d_{j,a}>1,\\
        1, & d_{j,a}=1,\\
        \infty, & d_{j,a}<1,
    \end{cases}
\end{equation}
indicating that while the weights of smaller distances increase infinitely, the weights of larger distances diminish, i.e., have no impact on the estimation process.
Furthermore, it can be observed from \eqref{eq:simplified} and \eqref{eq:limit_to_infinity} that the influence on the estimation process becomes increasingly concentrated at even closer distances as $n$ increases; in other words, the range of close distances narrows with growing $n$.
\begin{prop}\label{prop:spatial}
The BSs in closer proximity have more correlated traffic patterns.
\end{prop}
\begin{proof}
This kind of statements usually require an empirical proof. 
In this regard, the authors in~\cite{spatial-proof} found that the traffic patterns of the closer BSs in terms of distance are correlated to each other. 
Similarly, the study in~\cite{spatial-2} demonstrated a spatial correlation in BS traffic patterns.
Additionally, the work in~\cite{spatial-milan-2}, which specifically analyzed the same dataset as our study, though not explicitly stated, suggests from their analyses that the traffic of closer cells exhibits greater similarity.
\end{proof}

Considering the results derived from \eqref{eq:simplified} and Proposition \ref{prop:spatial}, it becomes evident that as the estimation process is increasingly influenced by cells closer to the cell-in-question (as indicated by the growing value of $n$), the estimation error, $\epsilon$, given by $\epsilon = |\hat\lambda_j - \lambda_j|=\left|\dfrac{{\sum\limits_{a = 1}^N {\lambda _a  \times\frac{{d_{\max } }}{{d_{j,a}^n }} } }}{{\sum\limits_{a = 1}^N {\frac{{d_{\max}}}{{d_{j,a}^n }} } }} - \lambda_j\right|$ tends to be lower, since the estimation has more correlated components.
\end{proof}

The parameter $n$ plays a crucial role in the estimation process, such that the increasing values of $n$ make the estimation more immune to the growing number of SBSs, $N$.

\begin{corollary}\label{cor:moreN}
 For a fixed value of $n$, the growing values of $N$ make the traffic load estimation worse, such that the estimation error, $\epsilon$ increases with $N$.   
\end{corollary}
\begin{proof}
As the number of SBSs, $N$, increases, more BSs that are farther from the cell-in-question contribute to the cell load estimation process, which, in turn, boosts the cell load estimation error, $\epsilon$.
\end{proof}

\begin{lemma}
The deviation in the estimation error with increasing values of $N$ reduces with the growing values of $n$.
\label{lemma1}
\end{lemma}
\begin{proof}

Corollary~\ref{cor:moreN} states that the increasing values of $N$ tend to increase the traffic load estimation error, $\epsilon$, such that $\dfrac{\text{d}\epsilon (N)}{\text{d}N}>0$.
On the other hand, Theorem~\ref{theorem2} indicates that the estimation error, $\epsilon$, reduces with growing values of $n$, i.e., $\dfrac{\text{d}\epsilon(n)}{\text{d}n}<0$.
When the effects of increasing values of $N$ and $n$ are combined, the effect of increasing $N$ is mitigated with the effect of increasing $n$, such that $\dfrac{\text{d}\epsilon (N)}{\text{d}N}\cdot \dfrac{\text{d}\epsilon(n)}{\text{d}n}<0$, and thereby if both $N$ and $n$ grow together, the deviation in the estimation error reduces.
\end{proof}

\subsection{Random Cell Selection Traffic Load Estimation}
This approach utilizes a random selection of surrounding cells for traffic load estimation. 
\subsubsection{Random Selection without Weighting}
The traffic load of a sleeping SBS is estimated based on the average traffic load of randomly selected surrounding SBSs. The estimation is calculated using a formula similar to \eqref{eq14}, where the selection of neighboring cells, $a$, is random.
\subsubsection{Random Selection with Weighting}
This variation applies a weighting mechanism to the randomly selected surrounding cells, the same as \eqref{eq15}, with the selection of $a$ being random. The weighting factors vary according to the distance from the sleeping SBS, enhancing the accuracy of the estimation. 
\subsection{Clustering-Based Traffic Load Estimation}
This approach involves clustering SBSs based on their traffic patterns and estimating the traffic load of a sleeping SBS using the average load of active SBSs within the same cluster. 
We employ the $k$-means algorithm, an unsupervised machine learning technique, for clustering the SBSs. 
The number of clusters, a crucial hyper-parameter in the $k$-means algorithm, is determined using the elbow method \cite{article}, which
assesses different cluster numbers by calculating the sum of squared errors (SSE) between data points and centroids for each potential cluster count. The SSE is given by \cite{article}
\begin{equation} 
SSE = \sum\limits_{g = 1}^G {\sum\limits_{x_m\in \kappa_g} {(x_m - \varkappa_g )} } ^2,
\label{eq18}
\end{equation}
where $G$ represents the optimal number of clusters, $\varkappa_g$ the centroid of each cluster $\kappa_g$, and $x_m$ each sample in $\kappa_g$. 
The optimal cluster number is identified at the point where the SSE curve forms an "elbow" before flattening.
\subsubsection{Multi-level Clustering-Based Traffic Load Estimation}
The multi-level clustering (MLC) approach involves repeated clustering of SBSs based on their traffic patterns to estimate the traffic load of offloaded SBSs. 
This method employs the elbow method to determine the optimum number of clusters, followed by the application of the $k$-means algorithm for clustering. The traffic load of a sleeping SBS is then estimated based on the average load of active SBSs in the same cluster. This iterative approach, as outlined in Algorithm \ref{alg1}, ensures progressively refined clustering with each layer, leading to more precise traffic load estimations for sleeping SBSs~\cite{maryamarxive}.

\begin{algorithm}
\caption{Multi-Level Clustering (MLC) using $k$-means}
\DontPrintSemicolon  
\SetAlgoLined  
\KwData{Traffic loads of SBSs $\lambda_{a}$, maximum number of layers $L$}
\KwResult{Clustered SBSs with estimated traffic loads}

\SetKwFunction{FMain}{MLC\_k\_means}
\SetKwProg{Fn}{Procedure}{:}{}
\Fn{\FMain{$\lambda$, $L$}}{
    Determine the optimal number of clusters $G$ using the elbow method\;
    Initialize layer count $l = 1$\;
    \While{$l \leq L$}{
        Perform $k$-means clustering on $\lambda$ to form $G$ clusters\;
        \For{cluster $\kappa_g$}{
            Calculate the mean traffic load $\mu_\text{m}$\;
            \For{sleeping SBS in $\kappa_g$}{
                Estimate the traffic load as $\mu_\text{m}$\;
            }
        }
        Update $\lambda$ with estimated ones for sleeping SBSs\;
        Increment the layer count $l$ by 1\;
    }
    \Return The final clusters with estimated traffic loads\;
}
\label{alg1}
\end{algorithm}

\begin{figure}[t]
\centerline{\includegraphics[width = 7.cm ]{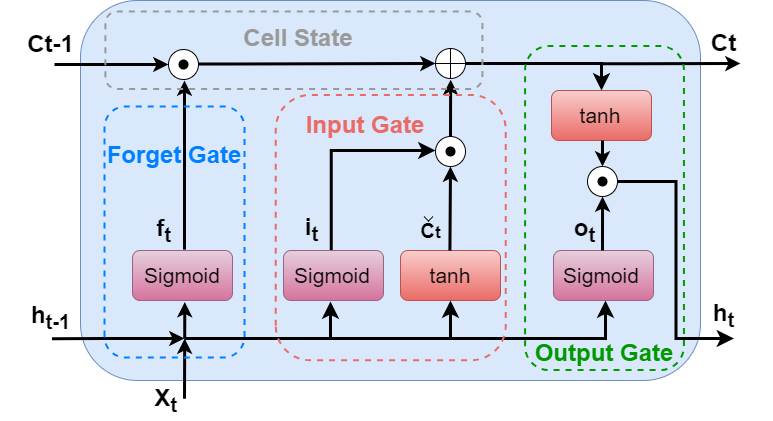}}
\caption{The architecture of an LSTM cell.}
\vspace{-.5cm}
\label{fig-6}

\end{figure}

\subsection{Temporal Traffic Load Prediction Using LSTM}
To improve the accuracy of traffic load estimation, we incorporate LSTM networks to model temporal dependencies in SBS traffic behavior. As a specialized type of recurrent neural network (RNN), LSTM is well-suited for capturing long-range sequential patterns, making it a natural fit for predicting time-evolving traffic profiles.
The architecture of an LSTM cell, illustrated in Fig.~\ref{fig-6}, includes three primary gates—forget, input, and output—that regulate the flow of information through a memory cell. 
The governing set of equations of the LSTM cell are given by
\begin{subequations}
\begin{align}
    f_t &= \sigma(W_f [h_{t-1}, x_t] + b_f), \label{eq:lstm_forget} \\
    i_t &= \sigma(W_i [h_{t-1}, x_t] + b_i), \label{eq:lstm_input} \\
    \tilde{C}_t &= \tanh(W_c [h_{t-1}, x_t] + b_c), \label{eq:lstm_candidate} \\
    C_t &= f_t \odot C_{t-1} + i_t \odot \tilde{C}_t, \label{eq:lstm_cell_state} \\
    o_t &= \sigma(W_o [h_{t-1}, x_t] + b_o), \label{eq:lstm_output} \\
    h_t &= o_t \odot \tanh(C_t),
    \label{eq:lstm_hidden}
\end{align}
\end{subequations}
where $x_t$ denotes the input vector at time step $t$, and $h_{t-1}$ is the hidden state from the previous time step.
The variables $f_t$, $i_t$, and $o_t$ represent the forget, input, and output gates, respectively. $\tilde{C}_t$ is the candidate cell state, and $C_t$ is the updated internal cell state (or memory). The output of the LSTM cell is the hidden state $h_t$, which is computed based on the current state of the cell and the output gate. 
The model parameters include trainable weight matrices $W$ and biases $b$. The function $\sigma(\cdot)$ denotes the sigmoid activation, and $\odot$ represents element-wise multiplication.

\subsubsection{Application of LSTM to Traffic Load Estimation}

In this study, LSTM networks are employed to estimate the current traffic loads of SBSs based on their historical traffic data. The data set consists of normalized traffic load values recorded every 10 minutes for a period of 30 days, resulting in 144 time slots per day. To enhance robustness, outliers are removed using z-score filtering with a threshold of 2.5. The filtered dataset is then randomly shuffled to improve generalization performance during training.
To prepare the data for supervised learning, a sliding window approach is applied. For each SBS, overlapping input-output sequences are generated, where each input sequence consists of a fixed number of past traffic load values—referred to as the window size—and the corresponding output is the immediate next time step. This formulation allows the model to capture short-term temporal dependencies in traffic fluctuations.
The data is split into two parts: 60\% for training and 40\% for testing.

The LSTM model is configured according to the hyperparameters summarized in Table~\ref{table:simulation_params}, including a single LSTM layer followed by a dense output layer. 
The model uses the Adam optimizer and is trained to minimize the mean absolute percentage error (MAPE). 
Its performance is evaluated by comparing predicted and actual traffic loads on the test dataset.

\section{Performance Evaluation}\label{sec:performance}
This section evaluates the effectiveness of our proposed traffic load estimation methods using the Milan dataset detailed in Section II-C. We assess both spatial and temporal estimation methods for predicting the traffic load of sleeping SBSs. Simulation parameters are provided in Table~\ref{table:simulation_params}.
For the spatial estimation methods, monthly traffic values are averaged per time slot for each SBS, and the algorithms are evaluated over multiple iterations with randomly selected sleeping SBSs to ensure robustness. 
In contrast, the LSTM-based temporal estimation model operates on the full traffic time series per SBS and is trained to estimate current traffic based on its own historical pattern.
\textit{Estimation error}, quantified using MAPE that captures the average relative difference between predicted and actual traffic loads, is used as a performance metric in the simulation campaigns.
\begin{centering}
\begin{table}[tb!] 
\caption{Simulation Parameters}
\centering 
\begin{tabular}{|c|c|} 
  \hline 
  \text{Parameter} & \text{Value} \\
  \hline
\multicolumn{2}{|c|}{\textbf{Spatial Estimation}} \\
\hline
Number of SBSs & 5000 \\

Number of time slots & 144 \\

Time slot duration & 10 m \\

Number of days & 30 \\

Number of iterations & 300 \\

Optimal $G$ using elbow method & 3 \\
\hline
\multicolumn{2}{|c|}{\textbf{Temporal Estimation}} \\
\hline
Learning rate & 0.001 \\


Number of LSTM layeres & 1 \\

Loss Function & MAE \\

Optimizer & Adam \\

Number of Epochs & 50 \\

Batch Size & 32 \\
\hline

\end{tabular}
\label{table:simulation_params} 
\end{table}
\end{centering}

Figure~\ref{fig-2} shows the MAPE for the distance-based estimation method with weighting, as a function of the number of neighboring SBSs.
The analysis, conducted for different powers of distance, \(n\), in \eqref{eq16} specifically $n=$ 1, 3, 5, and 10, provides key insights that are consistent with our theoretical findings:
\begin{itemize}
\item As noted in Corollary~\ref{cor:moreN}, increasing the number of neighboring SBSs ($N$) increases the estimation error.
\item Consistent with Theorem~\ref{theorem2}, higher values of $n$ significantly reduce the error. For example, the MAPE drops from 45\% at $n=1$ to about 15\% at $n=5$, and becomes minimal at $n=10$, underscoring the benefits of prioritizing proximity in weight calculations.
\item This behavior supports Proposition~\ref{prop:spatial}, affirming that traffic loads from closer BSs are more correlated and thus provide more reliable estimations.
\end{itemize}

\begin{figure}[t]
\centerline{\includegraphics[width = 7.cm ]{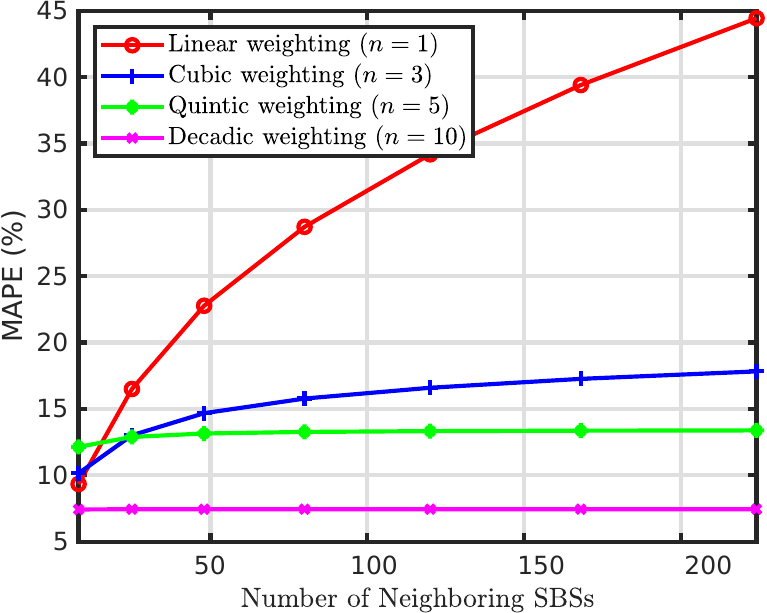}}
\caption{MAPE of distance-based estimation with weighting across different values of $n$ and number of neighbors.}
\label{fig-2}
 \vspace{-1em}
\end{figure}

Figure~\ref{fig-3} compares the performance of multiple spatial estimation techniques in terms of MAPE. The multi-level $k$-means clustering approach shows substantial improvement as the number of clustering layers increases, achieving near-zero error with seven layers. The distance-based method without weighting exhibits increasing MAPE as the number of neighboring cells increases, emphasizing the relevance of incorporating distance in estimation, as suggested by Theorem \ref{theorem2} and Proposition \ref{prop:spatial}, which highlight the importance of closer proximities in reducing errors.
In contrast, the random selection method with weighting sees decreased errors due to prioritizing closer neighbors, similar to the effects described in Theorem \ref{theorem2}. The random selection method without weighting is excluded due to its consistently high and unstable errors, underscoring the importance of distance consideration in neighbor selection as reinforced by Lemma~\ref{lemma1}.

\begin{figure}[t]
\centerline{\includegraphics[width = 7.cm ]{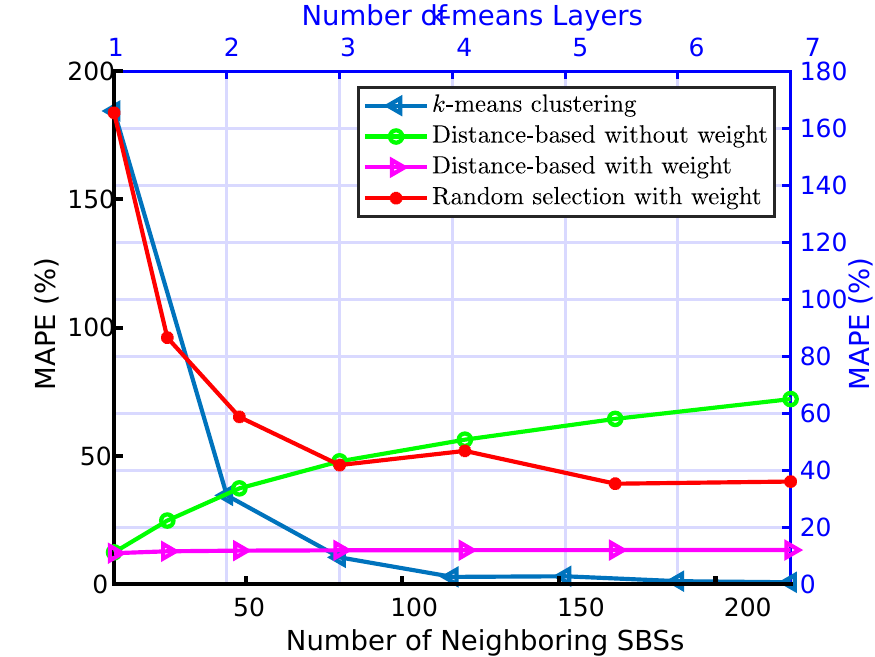}}
\caption{MAPE comparison of spatial estimation methods. The blue axis corresponds to clustering-based methods; the black axis corresponds to others.}
\label{fig-3}
\vspace{-1em}
\end{figure}

Figure~\ref{fig-7} presents the MAPE results for the LSTM-based temporal estimation method under various configurations of window size (i.e., number of past time steps) and number of LSTM units.
As shown in the figure, increasing the window size improves prediction accuracy, particularly when paired with a sufficient number of LSTM units. For instance, with 5 LSTM units, the MAPE decreases from 4.17\% at a window size of 4 to 1.22\% at a window size of 12.  
This trend becomes more consistent with 10 and 20 LSTM units, where the MAPE reaches as low as 0.68\% and 0.64\%, respectively.
These results confirm that LSTM-based modeling is effective in capturing sequential patterns in SBS traffic. Furthermore, they underscore the importance of tuning model hyperparameters—such as input window size and LSTM capacity—to achieve robust and accurate predictions under varying network conditions.

\begin{figure}[h!]
\centerline{\includegraphics[width = 7.cm ]{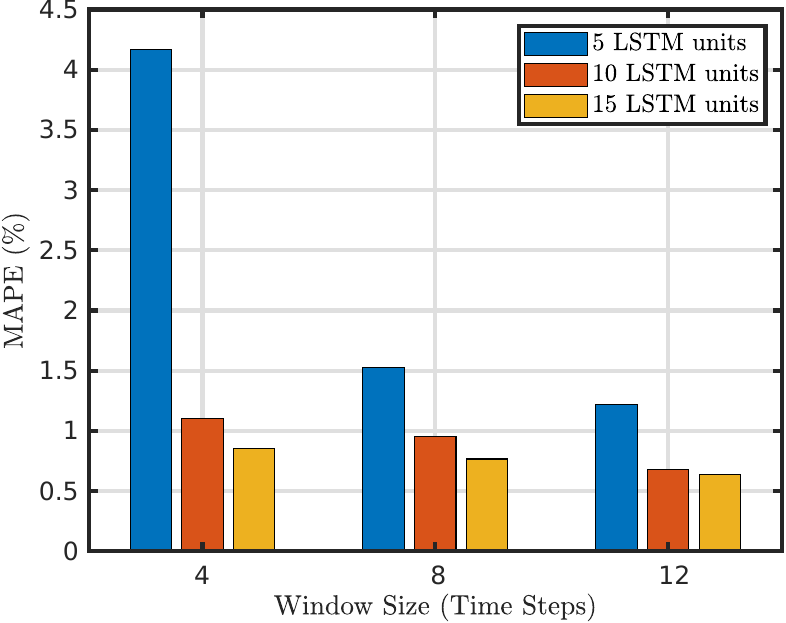}}
\caption{MAPE of LSTM-based traffic load estimation for varying window sizes and LSTM units.}
\label{fig-7}
\vspace{-1em}
\end{figure}

\section{Conclusion}\label{sec:conclusion}

This study addresses the critical challenge of estimating traffic loads for SBSs in sleep mode within vHetNets—an essential step toward enabling practical and energy-efficient cell switching.  
We developed mathematical frameworks to characterize spatial estimation techniques and evaluated both spatial and temporal methods to bridge the traffic load information gap.
When properly configured, both multi-level $k$-means clustering and LSTM networks achieved a MAPE of less than 1\%, demonstrating high accuracy in traffic prediction.  
Validated using a real-world dataset from Milan, these methods offer robust solutions to a key implementation bottleneck in advanced cell switching schemes.
While this work primarily focuses on traffic load estimation accuracy, which is crucial to make cell switching techniques applicable in real life implementations, future research will investigate the practical implications of estimation errors on network power consumption and decision changes in cell switching strategies.

\ifCLASSOPTIONcaptionsoff
  \newpage
\fi
\small
\input{main.bbl}


\end{document}